\documentclass{article}

\usepackage{arxiv}

\usepackage[utf8]{inputenc} 
\usepackage[T1]{fontenc}    
\usepackage{hyperref}       
\usepackage{url}            
\usepackage{booktabs}       
\usepackage{nicefrac}       
\usepackage{microtype}      
\usepackage{lipsum}		
\usepackage{graphicx}
\usepackage[numbers]{natbib}
\usepackage{doi}

\usepackage{amsfonts, amscd, amssymb, amsthm, amsmath}
\usepackage{cleveref}

\crefname{secinapp}{appendix}{appendices}
\Crefname{secinapp}{Appendix}{Appendices}

\def\cD{\mathcal{D}}\def\cJ{\mathcal{J}}\def\cM{\mathcal{M}}\def\cN{\mathcal{N}}\def\cS{\mathcal{S}}

    \def\EE{\mathbb{E}}             \def\RR{\mathbb{R}}

\def\zzero{\mathbf{0}}

\def\<{\langle} \def\>{\rangle}
\def\({(\!(}\def\){)\!)}

\usepackage{algorithm}
\usepackage{algorithmic}

\theoremstyle{plain} \newtheorem{thm}{Theorem}   \newtheorem{cor}[thm]{Corollary} 
\theoremstyle{definition} \newtheorem{defn}{Definition}  \newtheorem{asmtn}{Assumption}

\newcommand\independent{\protect\mathpalette{\protect\independenT}{\perp}}
\def\independenT#1#2{\mathrel{\rlap{$#1#2$}\mkern2mu{#1#2}}}

\newcommand{\eps}{\varepsilon}
\newcommand{\Lap}{\text{Lap}}

\title{Differentially Private Average Treatment Effect Estimation by Propensity Score Blocking}

\author{{Duncan Stewardson}\\
	Department of Computer Science\\
	Reed College\\
	Portland, OR, 97202 \\
	\texttt{dstewardson@reed.edu} \\
	\And
	{Grayson W. White} \\
	Department of Mathematics and Statistics\\
	Reed College\\
	Portland, OR, 97202 \\
	\texttt{gwhite@reed.edu} \\
    \AND
	{Adam Groce}\\
	Department of Computer Science \\
    Reed College \\
	Portland, OR, 97202 \\
	\texttt{agroce@reed.edu} \\
}

\date{}

\renewcommand{\shorttitle}{Differentially Private ATE by Propensity Blocking}

\hypersetup{
pdftitle={Differentially Private Average Treatment Effect Estimation by Propensity Score Blocking},
pdfauthor={Duncan Stewardson, Grayson W. White, Adam Groce},
pdfkeywords={Differential Privacy, Causal Inference, Average Treatment Effect},
}

\begin{document}

\maketitle

\begin{abstract}
Average treatment effect (ATE) estimation in observational studies is a fundamental statistical tool used frequently in social science, medicine, and other fields.  These fields often work with sensitive data where privacy protections are important, so a differentially private mechanism for ATE estimation is highly desirable.  Here we present two propensity score-based algorithms for ATE estimation on observational data, one improving the inverse probability weighting (IPW) method used in prior work, and the other using blocking on the propensity score (BPS).  Both show lower error and less bias than prior work, with the BPS-based algorithm frequently reducing error by 75\% or more compared to prior work.
\end{abstract}

\section{Introduction}

In fields ranging from economics to medicine and beyond, data analysts are looking to estimate the effect of some event or intervention, often referred to as a ``treatment'', on an outcome of interest. Motivated by questions such as ``Does a new medicine decrease heart disease?'', ``Does a new teaching method increase test scores?'', or ``Does a change in public policy reduce crime?'', analysts must implement methods to come to a statistically-defensible answer. When feasible, analysts will use randomized controlled trials that guarantee any difference in the observed outcome is due only to the treatment and sampling variability.  But often a randomized controlled trial is expensive or unethical, and analysts instead turn to the vast quantity of available observational data.  In this case, the analyst must account for the fact that the treatment was not assigned randomly and that, in the observed sample, certain groups may have been more likely to receive treatment than others. For example, perhaps those with worse underlying heart health were more likely to be given a medicine. Since the severity of one's heart condition may be related to, or ``confounded with'', the effect the medicine we must use methods that account for this relationship to estimate the average treatment effect (ATE) on our population of interest.  

One common way to account for these confounding factors in the assignment of treatment is the use of propensity scores.  Here the analyst estimates each individual's probability of treatment and uses those estimates to stratify, match, or weight the data points \cite{imbens_nonparametric_2004}.

The data sets needed for these studies often include a variety of highly sensitive (or legally-protected) attributes, and the studies often rely on existing data, rather than individuals who are opting into participation in a research study.  This makes the protection of privacy particularly important.  Our goal is to provide differentially private methods for ATE estimation.  Differential privacy \cite{dwork_calibrating_2006} provably guarantees that an individual's presence in the data does not result in leakage of sensitive information.  While differential privacy research often targets large data sets, we are interested in methods that can also work at the small sample sizes to which ATE is often applied.  We are also interested in concrete performance, not asymptotic results.

\paragraph{Contributions.}  This work presents two algorithms for nonparametric estimation of the ATE in observational studies.  In order to estimate ATE, we make no strict assumptions about the data and only require the analyst to place approximately accurate bounds on the potential values of the covariates and outcomes.  Our first algorithm follows prior work in using \textit{inverse probability weighting} (IPW) \cite{rosenbaum_model-based_1987}, though we improve on that work in several ways.  Our second algorithm instead uses {\it blocking-on-the-propensity-score} (BPS) \cite{rosenbaum_central_1983}, which to our knowledge has not been previously studied in the private setting. We then evaluate both algorithms on both simulated and real data.  Both algorithms have better accuracy and bias than prior work, with the BPS-based algorithm showing the best results in almost all settings and often reducing error by 75\% or more. We also provide a python package that implements all the methods described in this paper\footnote{https://pypi.org/project/dpate/}.

\section{Background}

\subsection{Differential privacy}
Differential privacy \cite{dwork_our_2006} provides a rigorous definition of privacy in the context of data queries. Informally, differential privacy guarantees that a single person's information has a minimal effect on the output distribution of a randomized mechanism. Let $\cD^n$ be the space of all $n\times d$ datasets. Then, differential privacy considers how a randomized mechanism behaves when run on {\it neighboring} datasets, and we consider two datasets, $D,D' \in \cD^n$, to be neighbors if they differ in the contents of one row. In other words, $D = \{x_1 \ldots x_i \ldots x_n\}$ and $D' = \{ x_1 \ldots x'_i \ldots x_n\}$.

\begin{defn}[$(\eps,\delta)$-differential privacy \cite{dwork_our_2006}]
    A randomized mechanism $\cM: \cD^n \to R$ satisfies $(\eps, \delta)$-differential privacy, or $(\eps, \delta)$-DP, if for all $\cS \subseteq R$ and all neighboring datasets $D, D'$:
    \[
    \Pr[\cM(D) \in \cS] \leq e^{\eps}\Pr[\cM(D') \in \cS] + \delta.
    \]
\end{defn}

When $\delta = 0$, we instead write $\eps$-DP. Our algorithms use $\eps$-DP, however prior work in DP ATE estimation uses the weaker $(\eps,\delta)$-DP.

A common way to achieve differential privacy is to compute some non-private function, and then add specifically calibrated noise to the output. Typically, either Laplacian or Gaussian noise is added, with the former achieving $\eps$-DP and the latter achieving $(\eps,\delta)$-DP. Gaussian noise can have higher error than Laplace noise. However, Gaussian noise can be preferred in certain cases due to having nicer properties under a large number of compositions. The noise must scale with the {\it sensitivity} of the function, which is the maximum difference in output that can occur by arbitrarily changing a single row in the dataset. 

\begin{defn}[The Laplace Distribution]
    The Laplace distribution, scaled by $b$, is defined by the probability density function:
    \[
    \Lap(x|b) = \frac{1}{2b}e^{-\frac{|x|}{b}}.
    \]
    We use $\Lap(b)$ to denote this distribution.
\end{defn}

\begin{defn}[$\ell_a$-Sensitivity \cite{dwork_calibrating_2006}]
    The $\ell_a$-sensitivity of a function, $f:D^n \to \RR^k$ is:
    \[
    \Delta_a(f)= \max\limits_{D,D'} ||f(D) - f(D')||_a
    \]
    for all neighboring $D, D'$. 
\end{defn}

\begin{defn}[The Laplace Mechanism \cite{dwork_calibrating_2006}]
    Given a function $f:\cD^n \to \RR^k$, the Laplace mechanism is defined as:
    \[
    \cM_L(D,f(\cdot),\eps) = f(D) + (Y_1,\ldots, Y_k),
    \]
    where $Y_i$ are i.i.d. random variables drawn from $\Lap(\Delta(f)/\eps)$.
\end{defn}

\begin{thm}[\cite{dwork_calibrating_2006}]\label{lap_mech}
    The Laplace mechanism satisfies $\eps$-DP.
\end{thm}

\begin{defn}[The Gaussian Distribution]
    The Gaussian distribution, with a center of $\mu$ and standard deviation $\sigma$, is defined by the probability density function:
    \[
    \frac{1}{\sqrt{2\pi\sigma^2}}e^{\frac{-(x-\mu)^2}{\sigma^2}}.
    \]
    We use $\cN(\mu,\sigma^2)$ to denote the Gaussian distribution.
\end{defn}

\begin{thm}[Gaussian Mechanism \cite{dwork_algorithmic_2013}] \label{gauss_mech}
    Output Perturbation using $\cN(0,\sigma^2)$ preserves $(\eps,\delta)$-differential privacy if $\sigma \geq \sqrt{2\ln(1.25/\delta)}\Delta_2(f)/\eps$.
\end{thm}

Differential privacy has some useful properties surrounding composition of functions, the first of which is that a private output can be arbitrarily processed without any privacy loss.

\begin{thm}[Post Processing \cite{dwork_algorithmic_2013}] \label{post-processing}
    Let $\cM: \cD^n \to R$ be some $(\eps,\delta)$-DP randomized mechanism, and let $f:R\to R'$ be an arbitrary mapping. Then $f\circ \cM$ is $(\eps,\delta)$-DP. 
\end{thm}

To run multiple queries one can use sequential composition (\Cref{seq_comp}) where the epsilon budget is split among queries. Alternatively, one can use parallel composition (\Cref{par_comp}) where the dataset is  split among queries.

\begin{thm}[Composition \cite{dwork_our_2006}] \label{seq_comp}
    Let $\cM = (\cM_1, \ldots, \cM_k)$ be a sequence of randomized mechanisms such that $\cM_i$ satisfies $(\eps_i,\delta_i)$-DP. Then, $\cM$ satisfies $(\sum_{i\in[k]} \eps_i, \sum_{i\in[k]} \delta_i)$-DP.
\end{thm}

\begin{thm}[Parallel Composition \cite{mcsherry2009privacy}] \label{par_comp}
    Let $D$ be a dataset partitioned into $k$ subsets: $\{D_i\}_{i\in[k]}$. Then, let $\{\cM_1,\ldots,\cM_k\}$ be a set of $k$ randomized mechanisms that each satisfy $(\eps,\delta)$-DP. Then $\cM = (\cM_1(D_1), \ldots, \cM_k(D_k))$ satisfies $(\eps,\delta)$-DP.
\end{thm}

\subsection{Average treatment effect}

For estimation of ATE, the standard model is the potential outcomes framework \cite{fisher_r_1923,splawa-neyman_application_1935, rubin_estimating_1974}. Under the potential outcomes framework, the dataset is defined as $D := \{(X_i, W_i, Y_i)\}_{i\in [n]}$, where $X_i \in \RR^{d}$ is a set of $d$ observed covariates, $W_i \in \{0,1\}$ is treatment assignment, and $Y_i \in \RR$ is some outcome. We consider $Y_i(0)$ to be the potential outcome for unit $i$ if $W_i = 0$, and similarly $Y_i(1)$ to be the potential outcome for unit $i$ if $W_i = 1$. This means:
\[
Y_i = Y_i(W_i)= \begin{cases}
                    Y_i(1) &\text{if } W_i = 1\\
                    Y_i(0) &\text{if } W_i = 0
                \end{cases}
\]

From here the ATE, denoted $\tau$, is defined to be the expected difference in these potential outcomes:
\[
\tau := \EE[Y(1) - Y(0)].
\]

The main difficulty in estimating $\tau$ comes from the fact that only one of the two potential outcomes is ever observed for any individual. One necessary assumption to be able to estimate $\tau$ is implied by how $Y_i$ is defined, but formally, it must be that the observed outcome of an individual is determined by their own treatment and not by the treatment assignment of any other individual. This is called the Stable Unit Treatment Value Assumption (SUTVA):

\begin{asmtn}[Stable Unit Treatment Value Assumption (SUTVA)]
    \citet{imbens_causal_2015} state SUTVA as follows: ``The potential outcomes for any unit do not vary with the treatments assigned to other units, and, for each unit, there are no different forms or versions of each treatment level, which lead to different potential outcomes.''
\end{asmtn}

With just this assumption, in a randomized experiment, it can be sufficient to simply take a difference in means of the treated and control groups to estimate $\tau$. However, in an observational study this is not sufficient as treatment assignment is not independent of the covariates. This means that there may be some confounding caused by these covariates. To be able to account for any confounding, we require two more assumptions. The first of which states that the observed covariates account for any and all confounding.

\begin{asmtn}[Unconfoundedness]
    The potential outcomes and the treatment assignment are independent conditioned on $X$. Formally, $(Y(0),Y(1)) \independent W | X$.
\end{asmtn}

\begin{asmtn}[Overlap]
    Each unit has a non-zero probability of each treatment assignment. Formally, $0<\Pr[W=1 | X] < 1$.
\end{asmtn}

We also use {\it overlap} to refer to how much these probabilities overlap, e.g. high overlap means there are many units with treatment probability close to 0.5. Further, this treatment probability is often referred to as the {\it propensity score}: $e(X) := \Pr[W=1 | X]$. The importance of the propensity score comes from the fact that we can reformulate the unconfoundedness assumption in terms of the propensity score: $(Y(0), Y(1)) \independent W | e(X)$ \cite{rosenbaum_central_1983}. However, in most observational studies, the true propensity score is not known, so it must be estimated. Logistic regression is often used to estimate the propensity score \cite{guo_propensity_2020}, and the estimated propensity score is denoted $\hat{e}(X)$.

\paragraph{Estimators for ATE.} Two basic estimators for ATE are \textit{Inverse Probability Weighting} (IPW) \cite{rosenbaum_model-based_1987} and \textit{Blocking-on-the-Propensity-Score} (BPS) \cite{rosenbaum_central_1983}. Both of these methods start with an estimated propensity score model, $\hat{e}(\cdot)$. The IPW estimator then weights the observed outcomes by the inverse of the estimated propensity score:
\[
\hat\tau_{\text{IPW}} = \frac{1}{n}\sum\limits_{i\in [n]}\frac{W_iY_i}{\hat{e}(X_i)} - \frac{(1-W_i)Y_i}{1-\hat{e}(X_i)}.
\]
Instead of weighting by the propensity score, the BPS estimator stratifies the data by propensity score, takes a simple difference in means for each block, and reports the average of these differences in means. Formally, partition the data into $m$ blocks, with $\cJ_k$ defined to be the set of indices in the $k$th block, $n_k$ to be $|\cJ_k|$, and $n_{k,w}$ to be the number of units in $\cJ_k$ with treatment assignment $w$. Then, we estimate $\tau$ using the following:
\[
\hat\tau_{\text{BPS}} = \frac{1}{n}\sum\limits_{k=1}^{m}\left[ \frac{n_k}{n_{k,1}}\sum\limits_{i\in\cJ_k}W_iY_i - \frac{n_k}{n_{k,0}}\sum\limits_{i\in\cJ_k}(1-W_i)Y_i\right].
\]
In most cases $m=5$ is sufficient to account for most confounding \cite{austin2011introduction, rosenbaum_central_1983}, and the blocks are then quintiles of the propensity score.

\subsection{Differentially private average treatment effect}

Literature on differentially private treatment effect estimation is sparse as well as new. Some prior work has considered differentially private ATE estimation in highly restrictive settings such as discrete covariates \cite{koga_differentially_2023}, binary outcomes \cite{ guha2025differentially, ohnishi_differentially_2025} or parametric estimators \cite{schroder_private_2025}. Some work has also considered a federated setting \cite{koga_differentially_2023, han2023multiply} where a central server must aggregate multiple private estimates from different treatment sites. Other work has considered the problem of {\it heterogeneous treatment effects} where the causal estimand is not constant across the population \cite{schroder_differentially_2025, niu_differentially_2022}.

To our knowledge, there are only two prior works in differentially private nonparametric estimation of ATE that allow for a wider range in outcomes and minimal assumptions about covariates. Lee et al. \cite{lee_privacy-preserving_2019} proposed splitting the dataset in two, using one portion to train a differentially private logistic regression model \cite{chaudhuri_differentially_2009} for the propensity score, and the other portion to estimate the ATE using IPW. Lebeda et al. \cite{lebeda_model_2025} proposed a subsample-and-aggregate \cite{nissim_smooth_2007} based method, where many non-private estimators for propensity score and outcome regression are trained and then privataely aggregated, with either the IPW, AIPW, or G-formula estimator being used. Both papers use the Gaussian mechanism, achieving $(\eps, \delta)$-DP.

\section{Proposed Algorithms}

In this section, we propose two algorithms for differentially private ATE estimation. Both algorithms presented in this section are in a general form, without explicit choices of parameters. See \Cref{experiments} for discussion surrounding choice of parameters. For both of these algorithms, we assume there is some publicly known value $C_y$ such that $|Y_i| \leq C_y$ for all $i \in [n]$. First, we present a sensitivity bound for $\hat\tau_{\text{IPW}}$:

\begin{thm} \label{IPW_sens}
    Suppose $|Y_i| \leq C_y$, and $\hat{e}(\cdot)$ is a fixed propensity score model with $\hat{e}(X_i) \in [\omega,1-\omega]$. Then, the $\ell_1$-sensitivity of $\hat\tau_{\text{IPW}}$ is $2C_y(n\omega)^{-1}$.
\end{thm}

\begin{proof}
	To start, we start by writing out the definition of sensitivity for $\hat\tau_{\text{IPW}}$:
	\begin{align*}
		\Delta(\hat\tau_{\text{IPW}}) &= \max\limits_{\substack{D,D'\\ d_H(D,D')=1}}\left|\left[\frac{1}{n}\sum\limits_{i=1}^{n} \frac{W_i Y_i}{e(X_i)} - \frac{(1-W_i)Y_i}{1-e(X_i)}\right] - \left[\frac{1}{n}\sum\limits_{i=1}^{n} \frac{W'_i Y'_i}{e(X'_i)} - \frac{(1-W'_i)Y'_i}{1-e(X'_i)}\right] \right|.
	\end{align*}
	We first note that since $e(\cdot)$ is assumed to be private, we can simplify this to only consider a single term in this summation and factor out the $\frac{1}{n}$:
	\begin{align*}
		\Delta(\hat\tau_{\text{IPW}}) &= \frac{1}{n}\max\limits_{\substack{\{X_i,W_i,Y_i\}, \{X'_i,W'_i,Y'_i\}}}\left|\frac{W_i Y_i}{e(X_i)} - \frac{(1-W_i)Y_i}{1-e(X_i)} - \frac{W'_i Y'_i}{e(X'_i)} + \frac{(1-W'_i)Y'_i}{1-e(X'_i)} \right|.
	\end{align*}
	At this point, we can break this down further to be the maximum over 2 possibilities, $W_i = W'_i$ and $W_i \neq W'_i$. Note that since $e(X_i), e(X'_i) \in [\omega,1-\omega]$, there is symmetry between $W_i = W'_i = 0$ and $W_i = W'_i = 1$ (and similar symmetry for the two possibilities of $W_i \neq W'_i$). Therefore, WLOG assume $W_i = 1$ (I omit $W_i, W'_i$ to declutter):
	\begin{align*}
		\Delta(\hat\tau_{\text{IPW}}) &= \frac{1}{n}\max\limits_{\substack{\{X_i,Y_i\}, \{X'_i,Y'_i\}}}\left\{\left|\frac{Y_i}{e(X_i)} - \frac{Y'_i}{e(X'_i)}\right|, \left|\frac{ Y_i}{e(X_i)} + \frac{Y'_i}{1-e(X'_i)} \right|\right\}.
	\end{align*}
	Now, we can see that the maximum value of the first term would be $Y_i = C$ and $Y'_i = -C$, with $e(X_i) = e(X'_i) = \omega$. For the second term, we get that the maximum value would be $Y_i= Y'_i = C$ and $e(X_i) = \omega, e(X'_i)=(1-\omega)$. This finally gets us:
	\begin{align*}
		\Delta(\hat\tau_{\text{IPW}}) &\leq \frac{1}{n}\max\left\{\left|\frac{C}{\omega} - \frac{-C}{\omega}\right|, \left|\frac{C}{\omega} + \frac{C}{1-(1-\omega)} \right|\right\}\\
		&= \frac{1}{n}\max\left\{\frac{2C}{\omega}, \frac{2C}{\omega}\right\}\\
		&= \frac{2C}{n\omega}.
	\end{align*}
\end{proof}

\begin{cor}
    The $\ell_2$-sensitivity of $\hat\tau_{\text{IPW}}$ is $\frac{2C}{n\omega}$.
\end{cor}

Our first algorithm (\textsf{SeqIPW}) uses the IPW estimator as in prior work \cite{lee_privacy-preserving_2019}. However, instead of splitting the dataset and relying on parallel composition, we train the propensity score model on the entire dataset, and again use the entire dataset for the final estimation of $\tau$. This requires us to use the sequential composition theorem (\Cref{seq_comp}) in our privacy analysis. Intuitively, by splitting the privacy budget as opposed to the dataset, we lose nothing in terms of privacy mechanisms since the added noise is inversely proportional to both $\eps$ (which is halved) and $n$ (which is doubled). However, by not splitting the dataset, we reduce the statistical error and bias that comes from having a smaller sample size for both the propensity score and the final estimation.

In addition to leveraging sequential composition, we use the Laplace mechanism instead of the Gaussian mechanism to achieve differential privacy for both the propensity score and final estimation. While the Gaussian mechanism can be more attractive in some settings, there are two main reasons for our use of Laplace noise. First, the Laplace mechanism achieves the strictly stronger $\eps$-DP as opposed to the Gaussian mechanism which achieves the weaker $(\eps,\delta)$-DP. Second, the variance of the noise added is lower with Laplace noise. (This is  experimentally confirmed in \Cref{app_gauss_lap}.) This second point means that simply changing from Gaussian noise to Laplace noise is an improvement to the prior work, which can be seen by experimental results in \Cref{app_gauss_lap} as well. Both \textsf{SeqIPW} and prior work clips the propensity score.  This is needed for privacy, but it is not purely a sacrifice for privacy, since it is also done sometimes in the non-private setting to reduce error \cite{imbens_nonparametric_2004}.

\begin{algorithm}
    \caption{\textsf{SeqIPW}} \label{SeqIPW}
    \textbf{Input}: Dataset $D = \{(X_i, W_i,Y_i)\}_{i\in[n]}$, privacy loss $\eps > 0$.\\
    \textbf{Parameters}: propensity score clip $\omega \in (0.0,0.5)$, privacy budget split $\alpha \in (0.0,1.0)$, and outcome bound $C_y$.\\
    \textbf{Output}: an $\eps$-DP estimate of $\tau$.
    \begin{algorithmic}[1]
        \STATE{Let $\eps_1 = \alpha\eps$, and $\eps_2 = (1-\alpha)\eps$.}
        \STATE{Train an $\eps_1$-DP propensity score model, $\hat{e}(\cdot)$, on $D$.}
        \STATE{Compute $\hat\tau_{\text{IPW}}$ on $D$ using $\hat{e}(\cdot)$ clipped to $[\omega, 1-\omega]$.}
        \RETURN{$\hat\tau_{\text{IPW}} + \Lap\left(\frac{2C_y}{n\omega\eps_2}\right)$}
    \end{algorithmic}
\end{algorithm}

\begin{thm}
    \textsf{SeqIPW} satisfies $\eps$-DP.
\end{thm}

\begin{proof}
    The propensity score model preserves $\eps_1$-DP by assumption. Then, by post-processing (\Cref{post-processing}) of $\hat{e}(\cdot)$ and the Laplace Mechanism (\Cref{lap_mech}) with sensitivity $\frac{2C_y}{n\omega\eps_2}$ (\Cref{IPW_sens}), $\hat{\tau}_{\text{IPW}}$ satisfies $\eps_2$-DP. Therefore, by composition (\Cref{seq_comp}), \Cref{SeqIPW} preserves $\eps_1 + \eps_2 = \eps$-DP.
\end{proof}

Our second algorithm (\textsf{DPBlocking}) uses the BPS estimator instead of IPW. The BPS estimator is attractive for two reasons. It can achieve lower bias and error than IPW in the non-private setting \cite{austin2011introduction, guo_propensity_2020}, and we can leverage the reasoning around differentially private histogram queries \cite{dwork_calibrating_2006} to privatize it with minimal noise. Specifically, if we were to naively compute $\hat\tau_{\text{BPS}}$, the sensitivity would be $4C_y$ since $\frac{n_k}{n_{k,w}}$ can be arbitrarily close to $n$ in the worst case. Instead, we can add noise to the bin counts as well as the sums of the outcomes in each bin which allows for less noise to be added overall. Another difference from the non-private literature is data-independent bin boundaries. Typically, quintiles of the propensity scores (or other data-dependent bin boundaries) are used instead of equal-spaced bins, which can help reduce variance due to having uneven bin sizes \cite{austin2011introduction}. Our algorithm simply uses equal-spaced bins, but it may be possible to optimize this by assuming some distribution of the propensity scores.

One consequence of adding noise to bin counts is that it may be possible for $n_{k,w}$ to be arbitrarily close to $0$, which could lead to a highly inaccurate estimate for $\tau$. To combat this, we check these noisy bin count values and if the bin count for either treatment group is less than some threshold $T$, we ignore that bin in the summation.

\begin{algorithm}
    \caption{\textsf{DPBlocking}} \label{DPBlocking}
    \textbf{Input}: Dataset $D = \{(X_i, W_i,Y_i)\}_{i\in[n]}$, privacy loss $\eps > 0$.\\
    \textbf{Parameters}: outcome bound $C_y$, privacy budget split $\alpha_1, \alpha_2, \alpha_3 \in (0.0,1.0)$ such that $\alpha_1+\alpha_2+\alpha_3 = 1$, number of bins $m$, and count threshold $T$.\\
    \textbf{Output}: an $\eps$-DP estimate of $\tau$.
    \begin{algorithmic}[1]
        \STATE{Let $\eps_1 = \alpha_1\eps$, $\eps_2 = \alpha_2\eps$, and $\eps_3 = \alpha_3\eps$.}
        \STATE{Train an $\eps_1$-DP logistic regression model, $\hat{e}(\cdot)$, on $D$.}
        \STATE{Set $\cJ_k = \left\{i \in [n] : \frac{k-1}{m} \leq \hat{e}(X_i) < \frac{k}{m}\right\}$ for $k\in [m]$.}
        \STATE{Set $\tilde{n} = 0$.}
        \FOR{$k = 1$ to $m$}
            \STATE{Compute $\tilde{n}_{k,w} = |\{i \in \cJ_k : W_i = w\}| + \Lap\left(\frac{2}{\eps_2}\right)$}
            \STATE{Compute $\tilde{S}_{k,w} = \Lap\left(\frac{2C_y}{\eps_3}\right) + \sum\limits_{\substack{i\in\cJ_k\\W_i = w}}Y_i$.}
            \STATE{Compute $\tilde{\tau}_k = (\tilde{n}_{k,1} + \tilde{n}_{k,0})\left(\frac{\tilde{S}_{k,1}}{\tilde{n}_{k,1}} - \frac{\tilde{S}_{k,0}}{\tilde{n}_{k,0}}\right)$}
            \IF{$\tilde{n}_{k,1} < T$ or $\tilde{n}_{k,0} < T$}
                \STATE{Set $\tilde{\tau}_{k} = 0$}
            \ELSE{}
                \STATE{$\tilde{n} = \tilde{n} + \tilde{n}_{k,1} + \tilde{n}_{k,0}$}
            \ENDIF{}
        \ENDFOR{}
        \STATE{Compute $\hat\tau_{\text{DPBlocking}} = \frac{1}{\tilde{n}}\sum\limits_{k\in [m]}\tilde{\tau}_k$}
        \RETURN{$\hat\tau_{\text{DPBlocking}}$}
    \end{algorithmic}
\end{algorithm}

\begin{thm}
    \textsf{DPBlocking} satisfies $\eps$-DP.
\end{thm}

\begin{proof}
    As before, the propensity score model preserves $\eps_1$-DP by assumption. Then, using it for binning is post-processing (\Cref{post-processing}). For the privatization of $\tilde{n}_{k,w}$ and $\tilde{S}_{k,w}$, we leverage the reasoning around sensitivity of histogram queries \cite{dwork_differential_2006}. Since a single row can alter at most two values of $\tilde{n}_{k,w}$ by at most 1, adding Laplace noise proportional to $\frac{2}{\eps_2}$ preserves $\eps_2$-DP. Similarly, with $\tilde{S}_{k,w}$, changing a single row either alters either one bin by at most $2C_y$ or two bins each by at most $C_y$. Therefore, adding Laplace noise proportional to $\frac{2C_y}{\eps_3}$ preserves $\eps_3$-DP. Then, the computation of $\hat{\tau}_\text{DPBlocking}$ is post-processing (\Cref{post-processing}). By composition (\Cref{seq_comp}), \Cref{DPBlocking} preserves $\eps_1+\eps_2+\eps_3 = \eps$-DP.
\end{proof}

\section{Experiments}\label{experiments}

\begin{figure*}[t]
    \centering
    \includegraphics[width=1\linewidth]{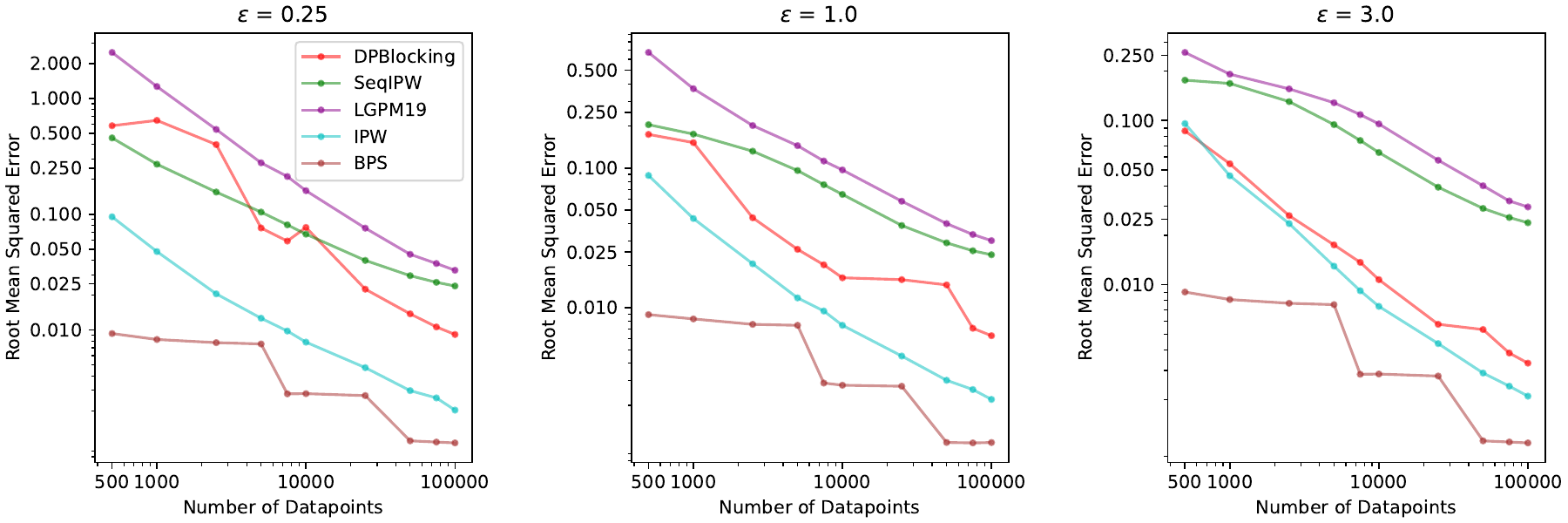}
    \caption{Observed error of both algorithms presented in this paper compared to Lee et al.'s algorithm (\textsf{LGPM19}) and two non-private estimations using $\hat\tau_{\text{IPW}}$ and $\hat\tau_{\text{BPS}}$. Both axes are log-scale.}
    \label{well-specified}
\end{figure*}

\subsection{Goals and Metrics}

We aim to estimate $\tau$ as closely as possible. The two main metrics that one can use for this process are observed bias and observed error. Specifically, over $S$ trials,we compute the following for each algorithm: 

\begin{align*}
\textit{RMSE}(\hat\tau)&= \left(\frac{1}{S}\sum_{s=1}^{S} (\hat{\tau}_s - \tau)^2\right)^{1/2} \\
\textit{Bias}(\hat\tau) &= \frac{1}{S}\sum_{s=1}^{S} \hat\tau_s - \tau,
\end{align*}

where $\hat\tau_s$ is the estimate of $\tau$ on the $s$th trial. While reducing both of these values is important, achieving low error captures the effect of bias as well as the variance of the estimator. Therefore, we mainly provide experimental results in terms of observed error. Specifically, we generate synthetic datasets with some known true $\tau$ value, which allows us to compute both the observed error and bias for our algorithms. We compare these values for observed error and bias to \citet{lee_privacy-preserving_2019}'s algorithm: \textsf{LGPM19}  and two non-private estimates using $\hat\tau_{\text{IPW}}$ and $\hat\tau_{\text{BPS}}$. For the non-private estimate using $\hat\tau_\text{BPS}$, we set the number of bins to 5 for $n \leq 5,000$, 10 for $5,000 < n \leq 25,000$, and 20 for $n > 25,000$. This was chosen with inspiration from \citet{lunceford2004stratification}, which suggested that 10 or 20 bins can achieve lower mean squared error for large $n$. For each experiment we calculate observed error and bias over $S = 500$ trials. In \Cref{app_lebj_comp}, we also compare to \citet{lebeda_model_2025}'s algorithm: \textsf{LEBJ25}. Further, note that \textsf{LGPM19} and \textsf{LEBJ25} use the weaker $(\eps,\delta)$-DP, so for the sake of consistency with prior work, we set $\delta = 10^{-6}$ in our experiments\footnote{All experimental code can be found at: https://github.com/dstewtes/DP-ATE-by-Blocking}.

\subsection{Choice of Parameters and Propensity Model}

This section outlines the parameter choices for both algorithms. We determined the choice of parameters by searching across a grid of reasonable values for each parameter on the well-specified data. In each case, we chose parameters that led to low root mean squared error. For full results of the experiments, see \Cref{paramter_tuning}.

For both algorithms, we use a logistic regression model for $\hat{e}(\cdot)$ implemented with Scikit-Learn \cite{scikit-learn} and privatize it by adding Laplace noise tuned to sensitivity $\frac{2}{n\eps_1}$ to the weights found by minimizing the binary cross-entropy loss with regularization constant of $1$ \cite{chaudhuri_differentially_2009}. We do this for consistency with prior work \cite{lee_privacy-preserving_2019} and due to common use of logistic regression in non-private literature \cite{guo_propensity_2020, austin2011introduction}, but our algorithms allow for any privately trained propensity score model to be used.

\paragraph{\textsf{SeqIPW}.} For \textsf{SeqIPW}, the main parameter to tune is $\alpha$ which determines how much of the privacy budget to use on the propensity score model. We found that $\alpha = 0.2$ consistently achieves low error. In addition to $\alpha$, there is the clipping parameter, $\omega$. There is no across-the-board winner for a value of $\omega$, as more clipping increases bias, especially when there is low overlap of the propensity scores (see \Cref{low_overlap}). However, we use $\omega = 0.1$ as it is low enough that the effect on bias is minimal, but high enough that the sensitivity is restrictively high.

\paragraph{\textsf{DPBlocking}.} For \textsf{DPBlocking}, we can vary the number of bins, $m$, the privacy budget parameters, $\alpha_1,\alpha_2,\alpha_3$, and the minimum bin count threshold $T$. We found that setting $T = 1.5, \alpha_1 = 0.1, \alpha_2 = 0.35,$ and $\alpha_3 = 0.55$ consistently achieves low error. No single choice of $m$ did well across all pairs of $n$ and $\eps$. In the non-private setting $m=5$ is typically sufficient since $5$ bins account for approximately $90\%$ of the bias introduced by confounding \cite{austin2011introduction}. However, when $n$ and $\eps$ are low, setting $m=5$ can lead to a large proportion of data points getting thrown out which leads to inaccurate outputs. Therefore, at low $n$ and low $\eps$ setting $m=3$ can achieve lower error. We found that adding an extra bin for every 2,000 data points, with $\eps$ determining when to start adding extra bins achieves low error, and we found that capping the possible number of bins based on $\eps$ can help reduce error at large $n$. Formally, we use the following equation for $m$, while capping $m$ at 5 or $20\eps$, whichever is higher:

\begin{align*}
    m = \max\left\{3, \left\lfloor\left(\frac{n}{1000} +\frac{\eps-0.25}{0.25}\right)/2\right\rfloor\right\}.
\end{align*}

\subsection{Data generating processes}

\paragraph{Well-Specified Setting.} The data generation process for the well-specified setting is heavily based on the generation process of \citet{kreif_regression-adjusted_2013}. We use 4 covariates, sampled from bivariate normal distributions:
\begin{align*}
    \begin{pmatrix}X_{i,1}\\X_{i,2}\end{pmatrix} &\sim \cN\left(\begin{pmatrix}2\\4\end{pmatrix}, \begin{pmatrix}1 & 0.2\\0.2 & 1\end{pmatrix}\right),\\
    \begin{pmatrix}X_{i,3}\\X_{i,4}\end{pmatrix} &\sim \cN\left(\begin{pmatrix}2\\4\end{pmatrix}, \begin{pmatrix}1 & 0.2\\0.2 & 1\end{pmatrix}\right).
\end{align*}
Then, to simulate treatment assignment, we sample from a Bernoulli distribution using the following logit:
\begin{align*}
\text{logit}[P(W=1)] &= 0.4 - X_1 + 0.5X_2 + 0.025X_2^2 - 0.25X_3 - 0.1X_4.
\end{align*}
Then, to simulate outcomes, we draw $Y_i \sim \cN(\mu_i,0.02)$, and translate then clip the outcomes such that $|Y_i| \leq \frac{\tau +1}{2}$. This gives us a true treatment effect of $\tau$. We set $\tau = 2.0$ for most experiments, and we define $\mu_i$ as the following:
\begin{align*}
    \mu_i = \tau W_i +0.1X_{i,1} - 0.05 X_{i,2} + 0.05X_{i,3} - 0.05X_{i,4}.
\end{align*}

For the purposes of scaling the covariates to the unit ball in a differentially private manner, we clip the covariates such that $|X_{i,1}|,| X_{i,3}| \leq 4.5$ and $|X_{i,2}|,| X_{i,4}| \leq 6.5$, then we scale $X_i$ by $1/||(4.5,6.5,4.5,6.5)||_2$.

\paragraph{Low-Overlap Setting.} To simulate a situation where there is low overlap in propensity scores, we use the same process for drawing $X_i$ and $Y_i$.  The only difference comes in the form of the logit for treatment assignment. Again, this logit is taken from the low-overlap setting of \citet{kreif_regression-adjusted_2013}:
\begin{align*}
\text{logit}[P(W=1)] = &1.5 - 2X_1 + X_2 + 0.05X_2^2 - 0.5X_3 - 0.2X_4.
\end{align*}

\paragraph{Unbalanced Setting.} To simulate a study where there are many more controls than treated, we alter the constant term of the propensity score logit to produce an unbalanced dataset. The data generating process is the same as the well-specified setting with the single difference being the propensity score logit, which we change to the following:
\begin{align*}
\text{logit}[P(W=1)] &= -1.5 - X_1 + 0.5X_2 + 0.025X_2^2 - 0.25X_3 - 0.1X_4.
\end{align*}
This formulation produced datasets where approximately one sixth of the units were treated and five sixths were controls.

\paragraph{Binary Outcomes.} To compare to prior works, we also test our algorithms on datasets with binary outcomes \cite{guha2025differentially}. This generating process also uses 4 covariates, and draws them as $X_i \sim \cN(\zzero, 0.8I + 0.2 J)$ where $J$ is a $4\times4$ matrix of $1$s. Then, treatment assignment is drawn from a Bernoulli distribution with the following logit:
\begin{align*}
    \text{logit}[P(W=1)] = {} &0.1 + 0.2\eta X_1 + 0.5 \eta X_2 - 0.25 \eta X_3 - 0.45 \eta X_4,
\end{align*}
with $\eta \in \{2,4\}$ to simulate high and low overlap respectively. Then, outcomes are similarly drawn from a Bernoulli distribution with logit:
\begin{align*}
    \text{logit}[P(Y(w)) = 1] = {} &0.15 -0.2X_1 + 0.3X_2 -0.4X_3 + 0.6X_4+\gamma w,
\end{align*}
with $\gamma\in\{0,1,2\}$ to simulate different levels of treatment effect. Note that with $\gamma = 0$, the true treatment effect is 0, $\gamma =1$ leads to true treatment effect of around $0.2$ and $\gamma = 2$ leads to true treatment effect of around $0.3$.

\subsection{Results}

\paragraph{Well-Specified Setting.} First, we present results using the well-specified data. For this experiment we tested the algorithms on a range of $n$ values from 500 to 100,000, and $\eps \in \{0.25,1.0,3.0\}$. As seen in \Cref{well-specified}, \textsf{SeqIPW} achieves the lowest error when both $n$ and $\eps$ are low. For larger $n$ and/or $\eps$, \textsf{DPBlocking} achieves the lowest observed error. Further, in \Cref{bias}, we see that \textsf{DPBlocking} achieves lower bias than \textsf{SeqIPW}, which is expected due to the bias of $\hat\tau_\text{IPW}$ with clipped propensity scores. Overall, \textsf{LGPM19} has both higher observed error and bias than both \textsf{SeqIPW} and \textsf{DPBlocking}. For $\eps = 0.25$ and $n=500$, we see that \textsf{DPBlocking}'s observed error was $0.58$, whereas \textsf{LGPM19} has observed error of $2.5$, which is even larger than the true value of $\tau$. When we increase $\eps$ to $1.0$, \textsf{DPBlocking} does even better, with observed error lower than $0.1$ for $n > 1{,}000$. Further, at $\eps = 3.0$, we see \textsf{DPBlocking} achieve error comparable to the non-private estimate using IPW. At $\eps = 0.25$, \textsf{DPBlocking} on average had a 63\% reduction in observed error compared to \textsf{LGPM19}; at $\eps = 1.0$, that increased to 75\%; and at $\eps = 3.0$, that increased even further to 84\%, with \textsf{DPBlocking} achieving a 90\% reduction in error at $n=25{,}000$ and $\eps = 3.0$. Further optimizing the choice of $m$ in \textsf{DPBlocking} would probably slightly improve accuracy further and would smooth the line in Figure 1.

\begin{figure}
    \centering
    \includegraphics[width=1\linewidth]{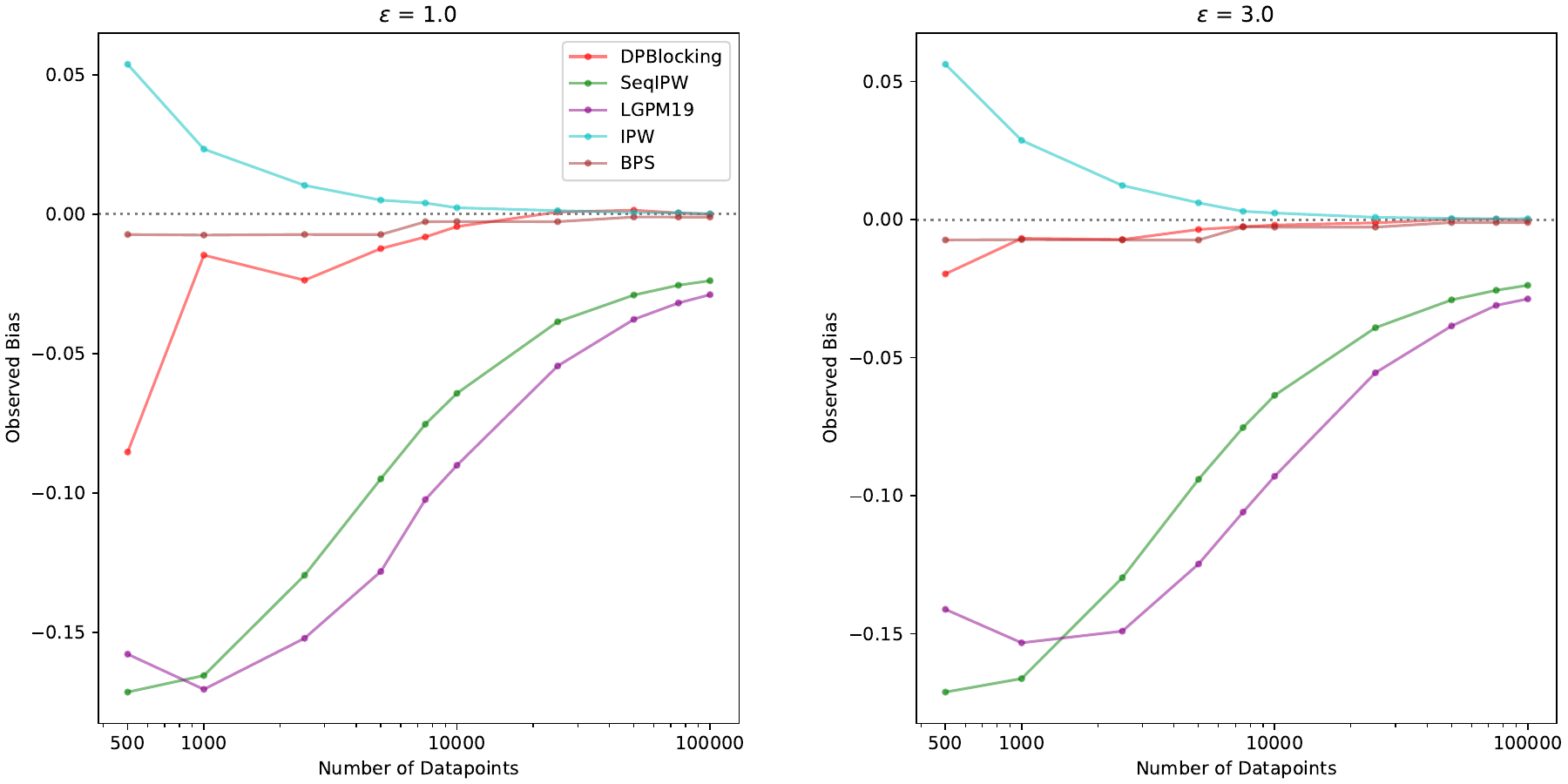}
    \caption{Observed bias of the both algorithms presented in this paper compared to \textsf{LGPM19} and two non-private estimations using $\hat\tau_{\text{IPW}}$ and $\hat\tau_{\text{BPS}}$. The $x$-axis is log-scale.}
    \label{bias}
\end{figure}

\paragraph{Unbalanced and Low Overlap.} For both the low overlap and unbalanced dataset setting, we range $n$ from $1{,}000$ to $10{,}000$, and use $\eps = 1.0$. In the low overlap setting (\Cref{low_overlap}), we expect any IPW-based method to suffer due to bias introduced when using IPW on data with low overlap. For all values of $n$, \textsf{DPBlocking} has lower observed error than both \textsf{SeqIPW} and \textsf{LGPM19}, and for $n>2000$, it even has lower error than the non-private estimate using IPW. In the unbalanced setting, we see that \textsf{DPBlocking} outperforms both private IPW-based methods, and for low $n$ also beats $\hat\tau_{\text{IPW}}$.

\begin{figure}
    \centering
    \includegraphics[width=1\linewidth]{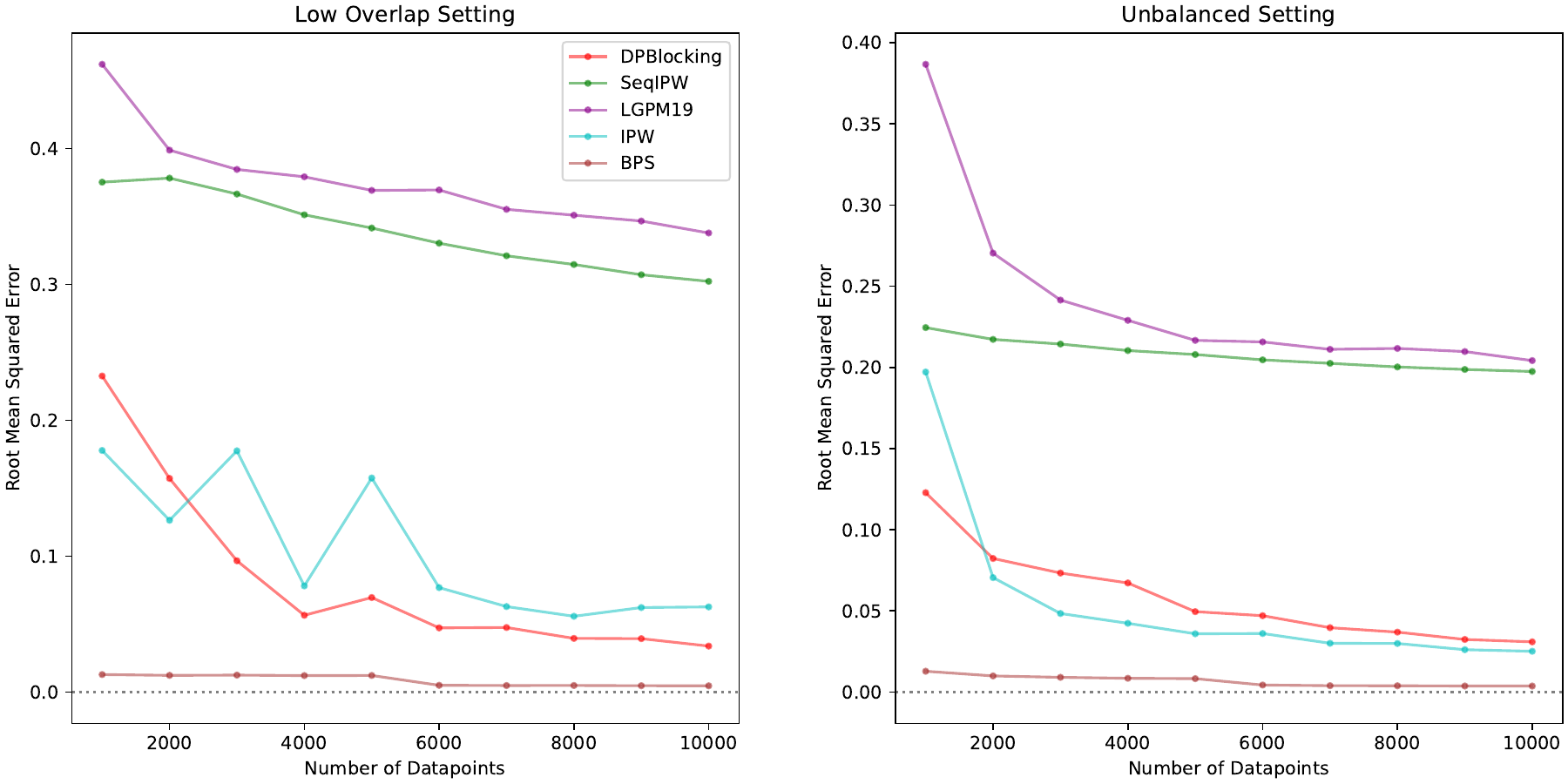}
    \caption{Observed error in the low overlap setting (left) and the unbalanced setting (right) of both algorithms presented in this paper compared to \textsf{LGPM19} and two non-private estimations using $\hat\tau_{\text{IPW}}$ and $\hat\tau_{\text{BPS}}$. $\eps = 1.0$ was used for both experiments.}
    \label{low_overlap}
\end{figure}

\paragraph{Binary Outcomes.}\Cref{binary_compare} presents a brief comparison to \citet{guha2025differentially}'s algorithm for ATE estimation under binary outcomes. The values found in the table are taken from Table 2 in \citet{guha2025differentially} as we were not able to reproduce their results. Note that for the sake of the privacy analysis of \textsf{DPBlocking}, we translate the outcomes to $\{-0.5,0.5\}$ instead of $\{0,1\}$. Across the board, \textsf{DPBlocking} gives more accurate point-estimates for $\tau$, but their algorithm also provides confidence intervals.

\begin{table}
    \centering
    \begin{tabular}{lcccccc}
        \hline
        \hline
         $\eta$ &  2 & 2 & 2 & 4 & 4 & 4\\
         $\gamma$ & 0 & 1 & 2 & 0 & 1 & 2 \\
         \hline
         \hline
         \textsf{GR24} & .016 & .016 & .015 & .023 & .021 & .024 \\
         \textsf{DPBlocking} & .013 & .012 & .012 & .017 & \textbf{.017} & \textbf{.016} \\
         \textsf{SeqIPW} & \textbf{.012} & \textbf{.011} & \textbf{.011} & \textbf{.013} & .026 & .035\\
         \hline
         \hline
    \end{tabular}
    \caption{Comparison of root mean squared error of both presented algorithms to Guha and Reiter's method \cite{guha2025differentially} for ATE estimation on binary outcomes. Values for \textsf{GR24} taken from Table 2 in their paper. This experiment uses $n = 10000$, and $\eps = 1.0$.}
    \label{binary_compare}
\end{table}

\subsection{National Supported Work Data}

In addition to simulation studies, we demonstrate the use of our algorithms on real data. We use \citet{dehejia1999causal}'s sample of \citet{lalonde1986evaluating}'s National Supported Work data. This dataset has been widely used for assessing ATE estimation, and we assess our methods on the six samples curated by \citet{dehejia1999causal}. For each of the six samples, the same set of 185 treated, and a varying number of controls were used.

The data has 8 covariates with 4 of which being binary. For the other four, since our methods assume bounded covariates, we clip age at 55 years old, education at 18 years, and both income in 1974 and income in 1975 at $75{,}000$ dollars. Then, the outcome for this data, $Y_i$, is the difference in income between 1978 and 1974 in tens of thousands of dollars, which we also clip to be between $75{,}000$ dollars and $-25{,}000$ dollars. We show results from 3 different datasets, each of which takes the treated units from NSW and combine that with a sample of controls from the Panel Study of Income Dynamics (PSID). PSID1 ($n=2675$) takes random samples from the entire PSID dataset, whereas PSID2 ($n=438$) and PSID3 ($n=313$) take samples specifically chosen to have similar covariates to the individuals in the NSW dataset.

\paragraph{Results.} To get a full picture of how our algorithms work on real data, we again run $S=500$ trials for each algorithm at epsilon values of 0.25, 1.0, and 3.0. Since this dataset does not have a known ground-truth value for $\tau$, we instead compute a non-private estimate of $\tau$ using IPW with normalized weights since normalizing the weights of IPW increases the accuracy of the estimator \cite{imbens_nonparametric_2004}. We use this estimate of $\tau$ as the ``true'' ATE. As seen in \Cref{real_rmse} and \Cref{real_bias}, for the largest dataset, \textsf{DPBlocking} has both the lowest observed error and bias for all values of $\eps$. However, on PSID2 and PSID3, the datasets are small enough that \textsf{SeqIPW} achieves lower error, but for PSID2, for $\eps$ of 1.0 and 3.0, \textsf{DPBlocking} has lower bias. In all cases, \textsf{DPBlocking} achieves lower error than \textsf{LGPM19}.  (See \Cref{stats} for formal hypothesis tests.)

\section{Conclusion}

In this paper we present two new algorithms for differentially private nonparametric estimation of the ATE. The first, \textsf{SeqIPW}, is a simple estimator based on the inverse probability weighting estimator. The second, \textsf{DPBlocking}, is an algorithm based on the 
BPS estimator that leverages the low sensitivity of histogram queries to add less noise to the output. We find that \textsf{DPBlocking} achieves lower observed error than prior work for all tested values of $n$ and $\eps$. We also find that \textsf{SeqIPW} achieves lower observed error in most cases. For most values of $n$ and $\eps$ both algorithms achieve lower observed bias than \textsf{LGPM19}. For low values of $n$ and $\eps$, we find that \textsf{SeqIPW} achieves lower error than \textsf{DPBlocking} when there are are no issues of low overlap or unbalanced datasets.  In all other settings \textsf{DPBlocking} achieves lower error and bias. 

\paragraph{Future Directions.} While point-estimates are an important tool in statistics, the generation of confidence intervals on top of these accurate point-estimates is an important future direction. Further, algorithms based on agumented IPW or IPW with normalized weights could potentially reduce bias \cite{austin2011introduction, imbens_nonparametric_2004}, but both are difficult to privatize.

\begin{table}
    \centering
    \begin{tabular}{c|c||c|c|c}
    \hline\hline
         $\eps$ & Method & PSID1 & PSID2 & PSID3  \\
         \hline\hline
         & $\hat\tau$ & 0.224 & 0.461 &  0.387\\
         \hline
        0.25& \textsf{DPBlocking} & \textbf{0.849} & 3.668 & 8.310\\
            & \textsf{SeqIPW} & 1.161 & \textbf{1.558} & \textbf{2.353}\\
            & \textsf{LGPM19} & 1.887 & 9.794 & 13.19\\
            \hline
        1.0& \textsf{DPBlocking} & \textbf{0.175} & 1.049  & 1.893\\
            & \textsf{SeqIPW} & 1.127 & \textbf{0.423} & \textbf{0.552}\\
            & \textsf{LGPM19} & 1.066 & 2.600 & 3.332\\
            \hline
        3.0& \textsf{DPBlocking} & \textbf{0.060} & 0.203 & 0.881 \\
            & \textsf{SeqIPW} & 1.125 & \textbf{0.198} & \textbf{0.212}\\
            & \textsf{LGPM19} & 1.017 & 0.916 & 1.191\\
            \hline\hline
    \end{tabular}
    \caption{Observed error on the NSW datasets based on a non-private estimate of $\tau$ using IPW with normalized weights}
    \label{real_rmse}
\end{table}

\begin{table}
    \centering
    \begin{tabular}{c|c||c|c|c}
    \hline\hline
         $\eps$ & Method & PSID1 & PSID2 & PSID3  \\
         \hline\hline
         & $\hat\tau$ & 0.224 & 0.461 &  0.387\\
         \hline
        0.25& \textsf{DPBlocking} & \textbf{-0.114} & 0.176 & 0.302\\
            & \textsf{SeqIPW} & 1.130 & \textbf{0.156} & \textbf{0.002}\\
            & \textsf{LGPM19} & 0.951 & 0.177 & 0.572\\
            \hline
        1.0& \textsf{DPBlocking} &\textbf{0.026}  &  \textbf{-0.027} & 0.333\\
            & \textsf{SeqIPW} & 1.124 &  0.146 & \textbf{0.025}\\
            & \textsf{LGPM19} & 0.769 & 0.168 & 0.152\\
            \hline
        3.0& \textsf{DPBlocking} & \textbf{-0.024} &  \textbf{-0.007} & 0.221 \\
            & \textsf{SeqIPW} & 1.125 & 0.142 & \textbf{0.005}\\
            & \textsf{LGPM19} & 0.729 & 0.203 & 0.023\\
            \hline\hline
    \end{tabular}
    \caption{Observed bias on the NSW datasets based on a non-private estimate of $\tau$ using IPW with normalized weights}
    \label{real_bias}
\end{table}

\bibliographystyle{abbrvnat}
\bibliography{references}

\newpage

\appendix

\section{Augmented Inverse Probability Weighting Estimator}

Another commonly used estimator for ATE is the augmented inverse probability weighting (AIPW) \cite{chernozhukov2018double}:
\begin{align*}
&\hat\tau_{AIPW} = \frac{1}{n}\sum\limits_{i\in [n]} \hat{g}_1(X_i) - \hat{g}_0(X_i) + \frac{W_i(Y_i-\hat{g}_1(X_i))}{\hat{e}(X_i)} - \frac{(1-W_i)(Y_i-\hat{g}_0(X_i))}{1-\hat{e}(X_i)},
\end{align*}
where $\hat{g}_w(X)$ is an estimate for $\EE[Y(w) |X]$. This estimator is attractive in the non-private setting due to it being double-robust  meaning that as long as either the propensity score model or the regression models are well-specified, the AIPW estimator is unbiased \cite{chernozhukov2018double}. Typically, $\hat{g}_0(X)$ and $\hat{g}_1(X)$ are linear regression models, but the choice of model is dependent on the distribution of the outcomes. For example, if the outcomes are known to be binary, logistic regression may be used instead. However, this estimator has a high sensitivity, which means the Laplace mechanism (or Gaussian mechanism) may be unattractive for privatization. Further, the need for the two outcome regression models poses an issue for any privatize method that requires the nuisance estimators to be private.

\section{Sensitivity of Estimators}

\subsection{Augmented Inverse Probability Weighting} \label[secinapp]{append_aipw_sens}

\begin{thm}
	Suppose $|Y_i| \leq C$, $e(X_i) \in [\omega,1-\omega]$ with $\omega \in (0,0.5)$  for all $i$, suppose $e(\cdot)$ is some fixed propensity score model, and suppose $g_{0}(\cdot)$ and $g_{1}(\cdot)$ are privately trained regression models clipped to $\pm C$. Then the $\ell_1$ sensitivity of $\hat\tau_{\text{AIPW}}$ is $\frac{4C}{n}\left(1+\frac{1}{\omega}\right)$. 
\end{thm}

\begin{proof}
    The $\ell_1$ sensitivity of $\hat\tau_\text{AIPW}$ is defined to be:
    \begin{align*}
        \Delta(\hat\tau_\text{AIPW}) &= \max\limits_{\substack{D,D'\\ d_H(D,D')=1}}\left|\left[\frac{1}{n}\sum\limits_{i=1}^{n} \frac{W_i (Y_i-g_1(X_i))}{e(X_i)} - \frac{(1-W_i)(Y_i-g_0(X_i))}{1-e(X_i)} + g_1(X_i) - g_0(X_i)\right]\right. \\
        &\qquad\qquad\left. - \left[\frac{1}{n}\sum\limits_{i=1}^{n} \frac{W'_i(Y'_i-g_1(X'_i))}{e(X'_i)} - \frac{(1-W'_i)(Y'_i - g_0(X'_i))}{1-e(X'_i)} + g_1(X'_i) - g_0(X'_i)\right] \right|.
    \end{align*}
    Since $e(\cdot)$ and $g_w(\cdot)$ are assume to be private, we can consider the change in just one value of $i$:
    \begin{align*}
        \Delta(\hat\tau_\text{AIPW}) &= \frac{1}{n}\max\limits_{\substack{\{X_i,W_i,Y_i\}, \{X'_i,W'_i,Y'_i\}}}\left|\left[\frac{W_i (Y_i-g_1(X_i))}{e(X_i)} - \frac{(1-W_i)(Y_i-g_0(X_i))}{1-e(X_i)} + g_1(X_i) - g_0(X_i)\right]\right. \\
        &\qquad\qquad\qquad\qquad\left. - \left[\frac{W'_i(Y'_i-g_1(X'_i))}{e(X'_i)} - \frac{(1-W'_i)(Y'_i - g_0(X'_i))}{1-e(X'_i)} + g_1(X'_i) - g_0(X'_i)\right] \right|.
    \end{align*}
    From here, we can consider either the case where treatment assignment changes and the case where treatment assignment does not change, this allows us to remove $W$ from the equation, further there is symmetry between $W_i=1$ and $W_i = 0$, so without loss of generality assume $W_i = 1$. Then we can substitute the bounds for the each variable to get the maximum:
    \begin{align*}
        \Delta(\hat\tau_\text{AIPW}) &= \frac{1}{n}\max\limits_{\substack{\{X_i,Y_i\}, \{X'_i,Y'_i\}}}\left\{\left|\left[\frac{(Y_i-g_1(X_i))}{e(X_i)}+ g_1(X_i) - g_0(X_i)\right] - \left[\frac{(Y'_i-g_1(X'_i))}{e(X'_i)} + g_1(X'_i) - g_0(X'_i)\right] \right|\right., \\
         &\qquad\qquad\qquad\quad\left. \left|\left[\frac{(Y_i-g_1(X_i))}{e(X_i)}+ g_1(X_i) - g_0(X_i)\right] - \left[- \frac{(Y'_i - g_0(X'_i))}{1-e(X'_i)} + g_1(X'_i) - g_0(X'_i)\right]\right|\right\}.\\
         &\leq \frac{1}{n} \max\left\{\left|\left[\frac{2C}{\omega} + 2C\right] - \left[\frac{-2C}{\omega} - 2C\right]\right|, \left|\left[\frac{2C}{\omega} + 2C\right] - \left[\frac{-2C}{1-(1-\omega)} - 2C\right]\right|\right\}\\
         &= \frac{4C}{n}\left(\frac{1}{\omega} +1\right)
    \end{align*}
\end{proof}

\begin{cor}
    The $\ell_2$-sensitivity of $\hat\tau_{\text{AIPW}}$ is $\frac{4C}{n}\left(\frac{1}{\omega} + 1\right)$.
\end{cor}

\subsection{Normalized Inverse Probability Weighting} \label[secinapp]{append_nipw_sens}

\begin{thm} \label{NIPW_sens}
	Suppose $|Y_i| \leq C$. Then, $\Delta(\hat\tau_{\text{NIPW}}) \leq 2C$. Where $\hat\tau_{\text{NIPW}}$ is $\hat\tau_{\text{IPW}}$ with normalized weights.
\end{thm}

\begin{proof}
    Simply, $\hat\tau_{\text{NIPW}}$ is a weighted mean with normalized weights, so $\hat\tau_\text{NIPW} \in [-C,C]$, which means that the maximum difference in output is bound above by $2C$. 
\end{proof}

\begin{cor}
    The $\ell_2$-sensitivity of $\hat\tau_{\text{NIPW}}$ is $2C$.
\end{cor}

\section{Other Tested Estimators}

In this section we present two other estimators for $\tau$ that can achieve low error for some specific circumstances. Both algorithms are generalizations of prior methods that make restrictive assumptions.

The first algorithm, \textsf{SeqNIPW} (\Cref{SeqNIPW}), is a generalization of the method presented by \citet{ohnishi_differentially_2025}. Their algorithm uses a specific propensity score model that is intractable for data with a high number of covariates, and they restrict the estimator to binary outcomes. Instead, we use logistic regression as in the other presented algorithms, and add noise that allows for any bounded outcomes.

\begin{algorithm}
    \caption{\textsf{SeqNIPW}} \label{SeqNIPW}
    \textbf{Input}: Dataset $D = \{(X_i, W_i,Y_i)\}_{i\in[n]}$, privacy loss $\eps > 0$.\\
    \textbf{Parameters}: propensity score clip $\omega \in (0.0,0.33)$, privacy budget splits $\alpha, \beta \in (0.0,1.0)$, and outcome bound $C_y$.\\
    \textbf{Output}: an $\eps$-DP estimate of $\tau$.
    \begin{algorithmic}[1]
        \STATE{Let $\eps_1 = \alpha\eps$, and $\eps_2 = (1-\alpha)\eps$.}
        \STATE{Train an $\eps_1$-DP propensity score model, $\hat{e}(\cdot)$, on $D$.}
        \STATE{Calculate $\hat{e}(X_i)$ clipped to $[\omega, 1-\omega]$ for each $i$.}
        \STATE{Compute:
        \[\hat\tau = \frac{\sum_{i=0}^{n}\frac{W_i\cdot Y_i}{\hat{e}(X_i)} + z_1}{\sum_{i=1}^{n}\frac{W_i}{\hat{e}(X_i)}+z_3} - \frac{\sum_{i=0}^{n}\frac{(1-W_i)\cdot Y_i}{1-\hat{e}(X_i)} + z_2}{\sum_{i=1}^{n}\frac{1-W_i}{1-\hat{e}(X_i)}+z_4},\]
        where $z_1,z_2 \sim \Lap\left(\frac{2C_y}{\omega\beta\eps_2}\right)$ and $z_3,z_4 \sim \Lap\left(\frac{2}{\omega(1-\beta)\eps_2}\right)$}
        \RETURN{$\hat\tau$}
    \end{algorithmic}
\end{algorithm}

\begin{thm}
    \Cref{SeqNIPW} satisfies $\eps$-DP.
\end{thm}

\begin{proof}
    The propensity score model preserves $\eps$-DP by hypothesis. Then, the use of $\hat{e}(\cdot)$ for estimation is private by post-processing. Then, to preserve privacy of the final output, we add calibrated noise to the numerators and denominators. Since a single datapoint can effect at most one numerator and one denominator, we can leverage the reasoning surrounding differentially private histogram queries \cite{dwork_differential_2006}. A single datapoint can either effect one numerator by $\frac{2C_y}{\omega}$ or each numerator by $\frac{C_y}{\omega}$, so adding noise drawn from $\Lap\left(\frac{2C_y}{\omega\beta\eps}\right)$ preserves $\beta\eps_2$-DP by the Laplace mechanism. Then, changing a single row can either effect one denominator by $\frac{1}{1-\omega} - \frac{1}{\omega}$ or each denominator by $\frac{1}{\omega}$, so adding noise drawn from $\Lap\left(\frac{2}{\omega(1-\beta)\eps}\right)$ preserves $(1-\beta)\eps_2$-DP by the Laplace mechanism. Then, the final calculation of $\hat\tau$ is private by post processing. Therefore, by composition, \Cref{SeqNIPW} preserves $\eps_1 + \beta\eps_2 + (1-\beta)\eps_2 = \eps$-DP.
\end{proof}

The second algorithm, \textsf{AggNIPW} (\Cref{AggNIPW}), is a generalization of the method presented by \citet{guha2025differentially}. Their method is restricted to binary outcomes, but we extend their method to allow for any bounded outcomes. This algorithm is a simple subsample-and-aggregate \cite{nissim_smooth_2007} based algorithm, where we take the noisy average of many non-private estimates for $\hat\tau$. One could further generalize this algorithm to allow for any $\hat\tau$ estimator, but we use $\hat\tau_{\text{NIPW}}$ due to its output being bounded by $C_y$.

\begin{algorithm}
    \caption{\textsf{AggNIPW}} \label{AggNIPW}
    \textbf{Input}: Dataset $D = \{(X_i, W_i,Y_i)\}_{i\in[n]}$, privacy loss $\eps > 0$.\\
    \textbf{Parameters}: number of aggregators $K$ and outcome bound $C_y$.\\
    \textbf{Output}: an $\eps$-DP estimate of $\tau$.
    \begin{algorithmic}[1]
        \STATE{Randomly partition $D$ into $K$ equal-sized subsets: $\{D_j\}_{j\in [K]}$}
        \FOR{$j = 1$ to $K$}
            \STATE{Train a propensity score model, $\hat{e}_j(\cdot)$ on $D_j$}
            \STATE{Compute $\hat\tau_{j} = \hat\tau_{\text{NIPW}}$ on $D_j$ using $\hat{e}_j(\cdot)$}
        \ENDFOR{}
        \RETURN{$\hat\tau = \frac{1}{K}\sum\limits_{j=1}^{K} \hat\tau_{j}+ \Lap\left(\frac{2C_y}{K}\right)$}
    \end{algorithmic}
\end{algorithm}

\begin{thm}
    \Cref{AggNIPW} satisfies $\eps$-DP.
\end{thm}

\begin{proof}
    By the Laplace mechanism with the sensitivity of $\hat\tau_{NIPW}$ (\Cref{NIPW_sens}), \Cref{AggNIPW} preserves $\eps$-DP.
\end{proof}

\section{Parameter Tuning} \label[secinapp]{paramter_tuning}

\subsection{\textsf{SeqIPW}}

For \textsf{SeqIPW}, the main parameter to tune is $\alpha$, and we test $\alpha \in \{0.2,0.3,0.4,0.5,0.6,0.7,0.8\}$. We tested these parameters using the well-specified data, over a range of $n$ and $\eps$. Results comparing root mean squared error can be seen in \Cref{seq_ipw_params_rmse}. Overall, we can see that setting $\alpha = 0.2$ or $\alpha = 0.3$ achieves low error across the board.

\begin{figure}
    \centering
    \includegraphics[width=\linewidth]{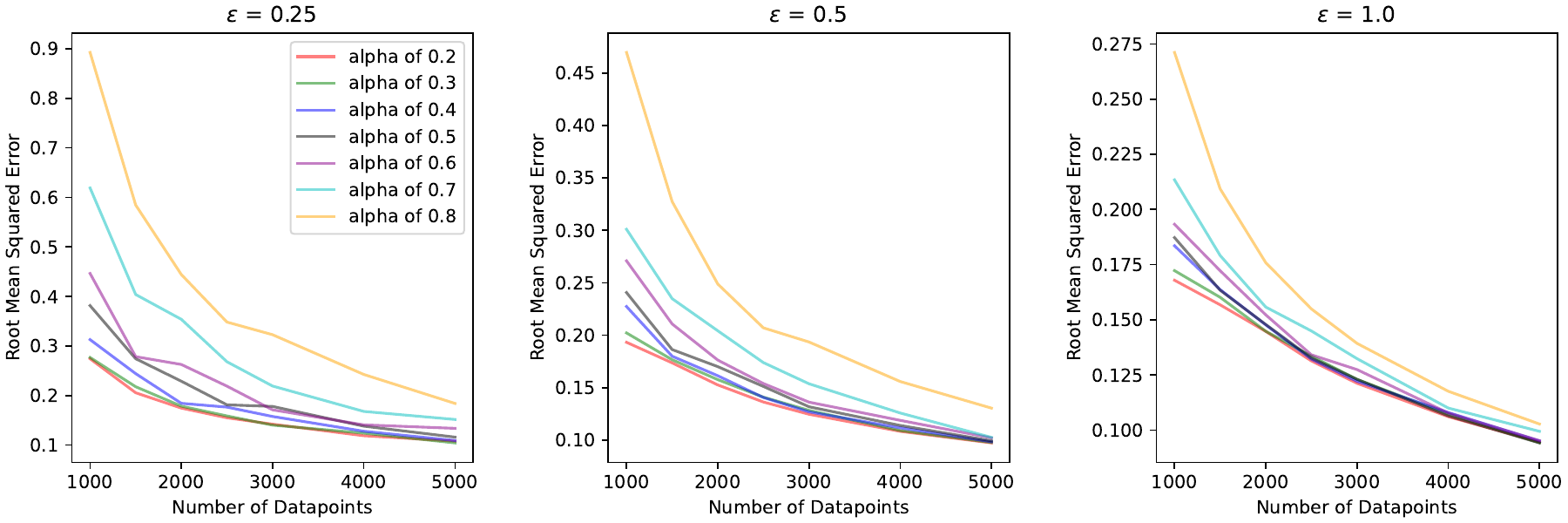}
    \caption{Experimental results using the well-specified data to test root mean squared error of \textsf{SeqIPW} with different choices of $\alpha$.}
    \label{seq_ipw_params_rmse}
\end{figure}

\subsection{\textsf{DPBlocking}}

For \textsf{DPBlocking}, there are five parameters to tune: $\alpha_1, \alpha_2, \alpha_3, T,$ and $m$. Given that we found lower budget going to the propensity score model results in lower error for \textsf{SeqIPW}, we decided to cut down the total space of the parameters by assuming $\alpha_1$ would be low, and $\alpha_3$ will likely be high. Specifically, for our first pass on optimiziong we tested
\begin{align*}
    \alpha_1 &\in \{0.1,0.2\}\\
    \alpha_2 &\in \{0.3,0.35,0.45\}\\
    T &\in \{0.5,1.0,1.5,2.0\}\\
    m &\in \{3,4,5,6,7,8,9,10\},
\end{align*}
with $\alpha_3 = 1- (\alpha_1+\alpha_2)$.  While this is certainly not an exhaustive search of parameters, as seen by results in the main body, even without the exhaustive search we achieve low error. 

We tested these parameters on values of $n$ from 1,000 to 10,000, and $\eps$ of 0.25, 1.0, and 3.0. Since there are too many parameters to reasonably create figures to show the results, we present a summary of findings in lieu of graphics. The first thing we note is that upon the first optimization run, we found that setting $\alpha_1 = 0.1$, $\alpha_2 = 0.35$, $\alpha_3 = 0.55$, and $T = 1.5$ achieves low error for all choices of $\eps$ and $n$.

There was no clear winner for $m$, but in general we saw a trend of adding more bins as $n$ and $\eps$ increased. While there was instantly obvious smooth formula for the choice of $m$, we found that:
\[
m = \max\left\{3, \left\lfloor\left(\frac{n}{1000} +\frac{\eps-0.25}{0.25}\right)/2\right\rfloor\right\}
\]
did a good job at achieving low error. However, when we tested this further on datasets larger than 10,000, we found that there was an increase in error for epsilon of 0.25 and 1.0. Therefore, we decided to cap the possible value of $m$ dependent on $\eps$. This left us with the following full formula for $m$:
\[
m = \min\left\{\max\left\{5,\lfloor \eps * 20\rfloor\right\},\max\left\{3, \left\lfloor\left(\frac{n}{1000} +\frac{\eps-0.25}{0.25}\right)/2\right\rfloor\right\}\right\}.
\]

\section{Further Experimental Results}

\subsection{Statistical Significance of Results} \label[secinapp]{stats}

In this section, we present results regarding the statistical significance of the experiment run on the NSW data in the main body of the paper. To test statistical significance, we used two-sided Wilcoxon signed-rank tests comparing the absolute error of each estimator. We used SciPy \cite{2020SciPy-NMeth}. We present the results as a series of tables reporting the p-values of these tests. In \Cref{real-data-significance}, we see that all results from this experiment are statistically significant for $\alpha = 0.01$.

\begin{table}
    \centering
    \begin{tabular}{ll||l|l|l}
        \hline\hline
        $\eps$ & dataset & \textsf{DPBlocking} V. \textsf{SeqIPW} & \textsf{DPBlocking} V. \textsf{LGPM19} & \textsf{SeqIPW} V. \textsf{LGPM19} \\ \hline\hline
        0.25 & psid & 6.21e-55 & 7.09e-41 & 3.55e-05 \\ 
         & psid2 & 1.43e-13 & 3.92e-62 & 3.85e-77 \\ 
         & psid3 & 1.63e-24 & 8.44e-44 & 3.99e-75 \\ \hline
        1.0 & psid & 1.26e-83 & 6.9e-83 & 1.5e-12 \\ 
         & psid2 & 2.87e-11 & 2.21e-66 & 1.28e-77 \\ 
         & psid3 & 3.82e-42 & 5.98e-34 & 2.57e-76 \\ \hline
        3.0 & psid & 1.26e-83 & 1.26e-83 & 4.88e-33 \\ 
         & psid2 & 0.0057 & 2.15e-72 & 4.14e-71 \\ 
        & psid3 & 8.32e-48 & 1.96e-19 & 1.26e-73 \\ \hline\hline
    \end{tabular}
    \caption{p-values of two-sided Wilcoxon signed-rank tests comparing absolute error of the algorithms on the NSW datasets.}
    \label{real-data-significance}
\end{table}

\subsection{Comparison to \textsf{LEBJ25}} \label[secinapp]{app_lebj_comp}

\begin{figure}
    \centering
    \includegraphics[width=1\linewidth]{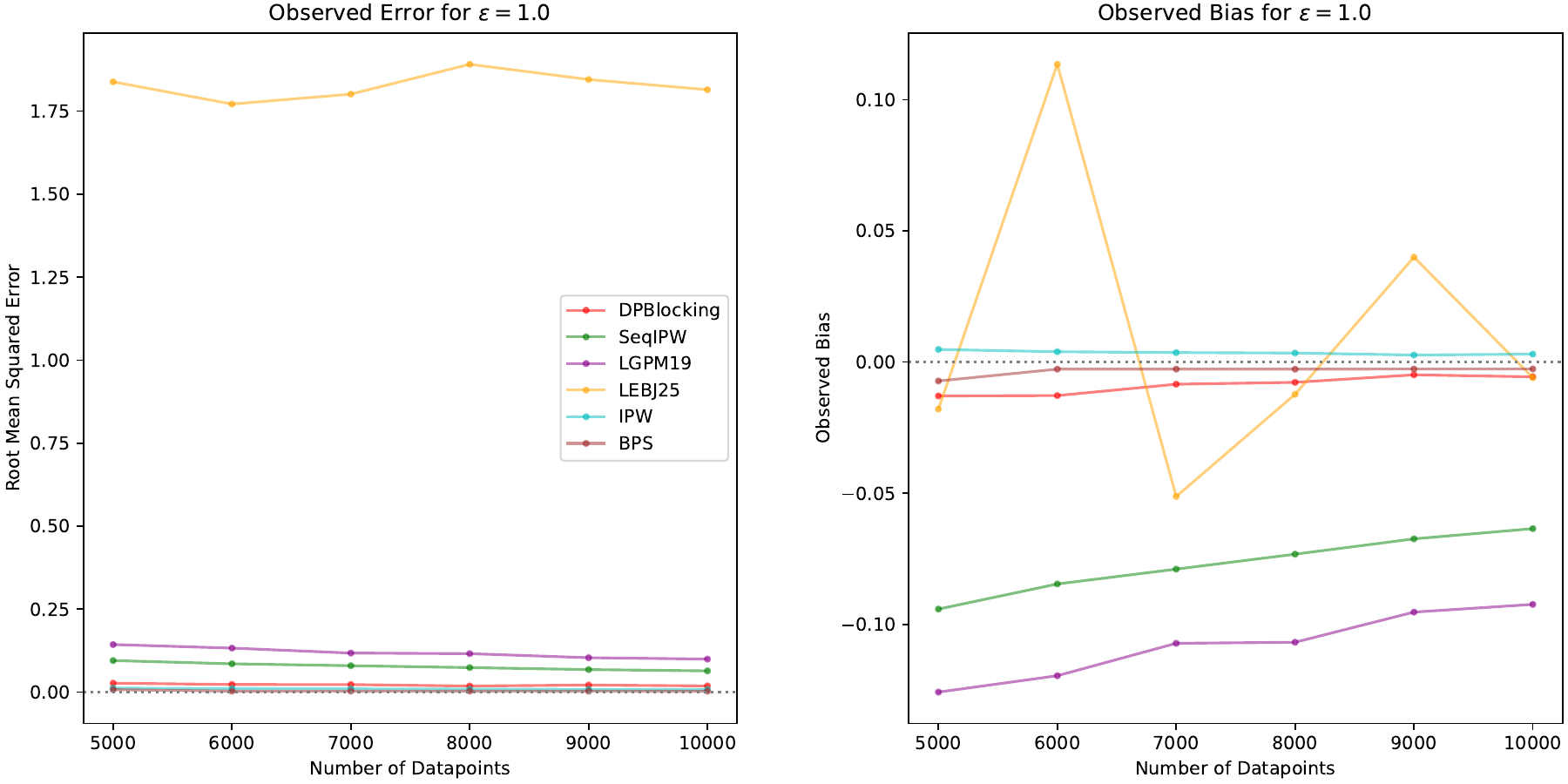}
    \caption{Comparison to \textsf{LEBJ25}. We show observed error (left) and bias (right) for $\eps = 1.0$ and a range of $n$ from 5,000 to 10,000 due to \textsf{LEBJ25} requiring large datasets since it is an aggregation-based method.}
    \label{LEBJ_compare}
\end{figure}

We also compared our results (as well as \textsf{LGPM19}) to the work of \citet{lebeda_model_2025} (which we denote \textsf{LEBJ25}). We used the well-specified data with $\eps$ of 1.0, and $n$ ranging from 5,000 to 10,000 due to \textsf{LEBJ25} sometimes throwing errors on smaller datasets due to not having enough samples in the small logisitic regression models. \textsf{LEBJ25} uses AIPW, which can achieve lower bias and error than IPW and BPS in the non-private setting \cite{chernozhukov2018double}. However, due to the higher sensitivity of AIPW as well as the sensitivity relying on $K$, AIPW has significantly higher error than the other tested methods as seen in \Cref{LEBJ_compare}. Further, while we ran 500 trials, we noticed that there is no clear pattern for the bias of \textsf{LEBJ25}, which might also be due to the greater noise added in this method. While it may be possible to optimize this method to increase $K$ with $n$ (possibly have $K$ grow with the square root of $n$), the results use $K = 200$ as they do in their paper, and we see that even at $n \leq 40{,}000$, the error is high suggesting that just setting $K = \sqrt{n}$ would not help for low $n$. 

\subsection{Gaussian v. Laplace} \label[secinapp]{app_gauss_lap}

One of the two optimizations we make to $\textsf{LGPM19}$ is the use of Laplace noise rather than Gaussian noise. To confirm that this is the valid change, we present results comparing the use of Laplace noise to the use of Gaussian noise for these algorithms. We use $\delta = 10^{-6}$ for all algorithms as that is the value that Lee et al. 2019 uses.

Results can be seen in \Cref{gauss_lap}. We use \textsf{SplitIPW} to denote \textsf{LGPM19} with Laplace noise instead of Gaussian noise. We can see that for every combination of $n$ and $\eps$, using Laplace noise achieves lower error than using Gaussian noise. This difference is especially apparent for lower values of $n$ and $\eps$. Further, we see that for low enough values of $n$ and $\eps$, simply switching to Laplace noise achieves lower than using Gaussian noise for \textsf{SeqIPW}, further showing the importance of using Laplace noise.

\begin{figure}
    \centering
    \includegraphics[width=1\linewidth]{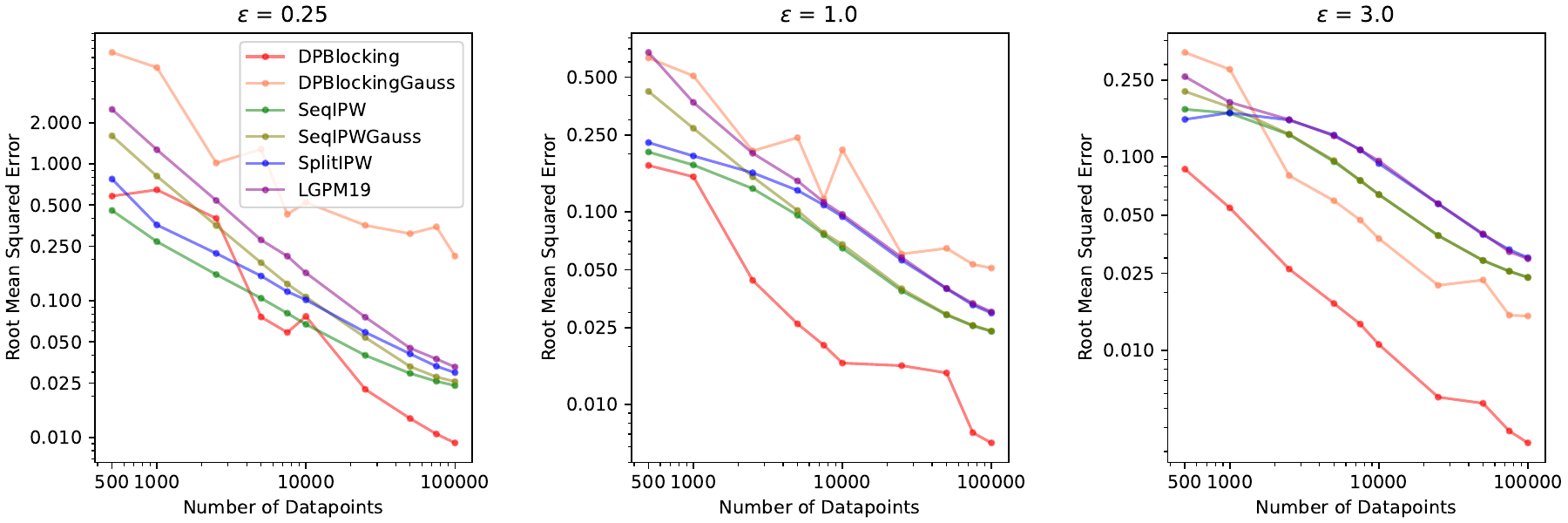}
    \caption{Observed error for algorithms using Gaussian versus Laplace noise.}
    \label{gauss_lap}
\end{figure}

\subsection{Other Values of $\eps$}

\begin{figure}
    \centering
    \includegraphics[width=1\linewidth]{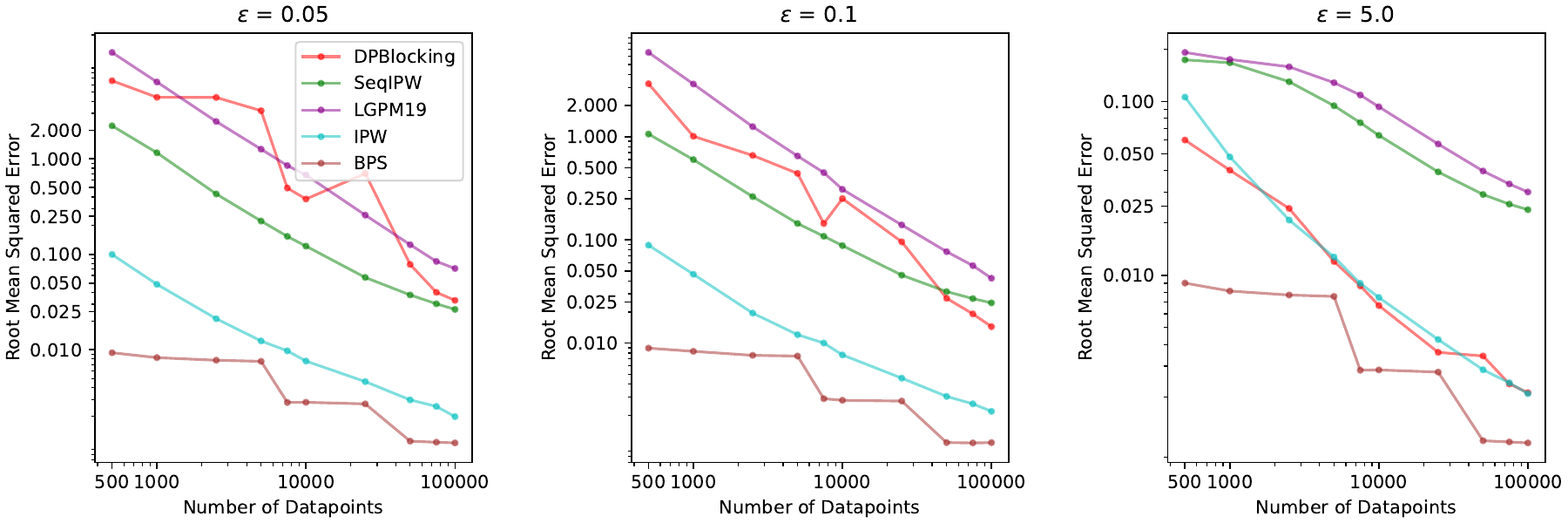}
    \caption{Observed error of presented algorithms tested on more values of $\eps$. Both axes are log-scale.}
    \label{more_eps}
\end{figure}

In this experiment, we tested our algorithms on $\eps$ values of 0.05, 0.1, and 5.0 using the well-specified data with $\tau$ of 2.0. As seen in \Cref{more_eps}, \textsf{SeqIPW} still beats \textsf{LGPM19} across the board, and for $\eps$ of 0.05 and 0.1, typically also beats \textsf{DPBlocking}. For $\eps$ of 0.05, \textsf{DPBlocking} can have high error, but at $\eps$ of 0.1, it still achieves lower error than \textsf{LGPM19}, and at $\eps$ of 5.0, it has error comparable to, if not lower than, the non-private IPW estimation.

\subsection{Other Estimators}

\begin{figure}
    \centering
    \includegraphics[width=1\linewidth]{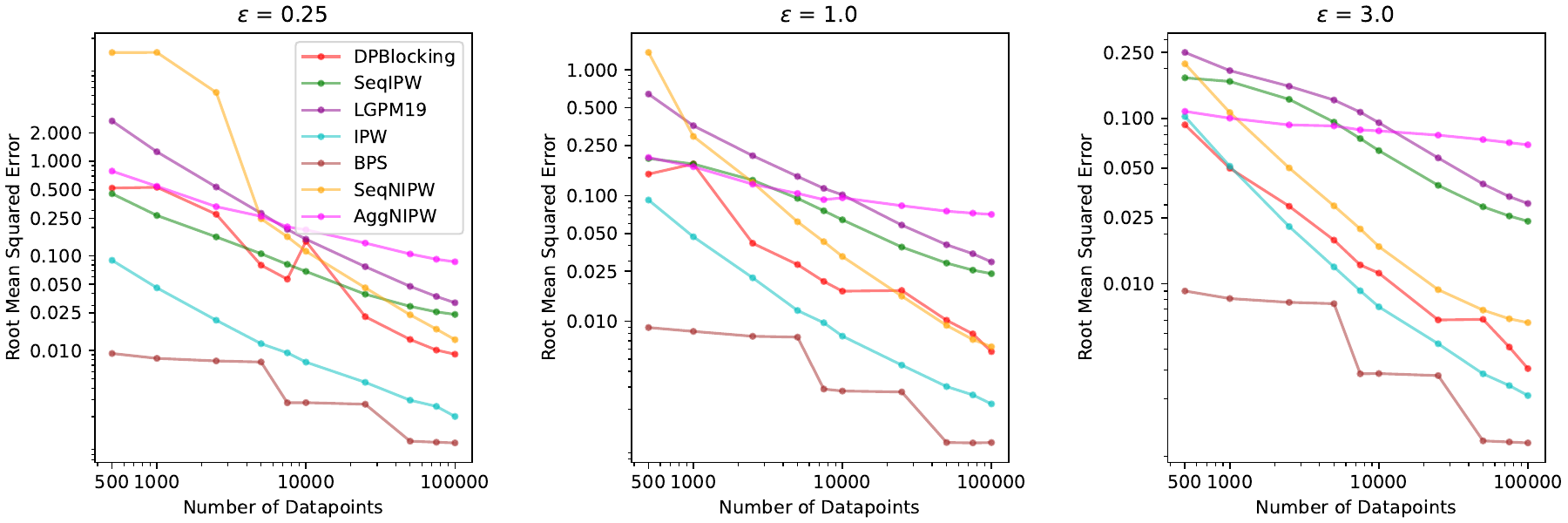}
    \caption{Observed error of presented algorithms alongside two generalizations of prior algorithms with restrictive outcome spaces. Both axes are log-scale.}
    \label{all_estim}
\end{figure}

In this experiment, we tested our algorithms and generalizations of \citet{guha2025differentially}'s and \citet{ohnishi_differentially_2025}'s algorithms. We used the well-specified data with $\eps$ of 0.25, 1.0, and 3.0, and a range of values of $n$ from 500 to 100,000. We see that \textsf{AggNIPW} beats \textsf{LGPM19} at roughly $n \leq 10{,}000$, but across the board \textsf{DPBlocking} still achieves lower observed error, and for all but $\eps = 3.0$, \textsf{SeqIPW} also beats \textsf{AggNIPW}. On the other hand \textsf{SeqNIPW} does well when the dataset is large and the epsilon value is high. We still find that \textsf{DPBlocking} achieves lower error in most cases, but around $n > 10{,}000$ and $\eps = 1.0$, \textsf{SeqNIPW} achieves similar error as \textsf{DPBlocking}. However, in the case of low $\eps$ and low $n$, \textsf{SeqNIPW} suffers from error much larger than the true value of $\tau$. We see that for $n \leq 2500$ and $\eps = 0.25$, \textsf{SeqNIPW} has observed error of $5$ or greater, when the value of $\tau$ is only 2.

\end{document}